\documentclass[sigconf]{acmart}

\renewcommand\footnotetextcopyrightpermission[1]{}
\usepackage{enumitem}
\usepackage{lipsum}
\usepackage{xspace}
\usepackage{xcolor}
\usepackage[linesnumbered,ruled]{algorithm2e}
\usepackage{amsbsy}
\usepackage{amsmath}
\usepackage{flushend}
\usepackage{pgfplots}
\usepackage{bbm}

\usepackage{natbib}
\newtheorem{example}{Example}
\newtheorem{theorem}{Theorem}

\newcommand{\sys}{\textsf{LikeMind}\xspace}
\newcommand{\separateshort}{\vspace{3pt}}

\AtBeginDocument{%
  \providecommand\BibTeX{{%
    \normalfont B\kern-0.5em{\scshape i\kern-0.25em b}\kern-0.8em\TeX}}}

\begin{document}

\title{Interactive and Explainable Point-of-Interest Recommendation using Look-alike Groups}

\author{Behrooz Omidvar-Tehrani, Sruthi Viswanathan, Jean-Michel Renders}
\email{firstname.lastname@naverlabs.com}
\affiliation{\institution{NAVER LABS Europe\\\textit{http://www.europe.naverlabs.com}}}
\renewcommand{\shortauthors}{Omidvar-Tehrani et al.}

\begin{abstract}
Recommending Points-of-Interest (POIs) is surfacing in many location-based applications. The literature contains personalized and socialized POI recommendation approaches which employ historical check-ins and social links to make recommendations. However these systems still lack customizability (incorporating session-based user interactions with the system) and contextuality (incorporating the situational context of the user), particularly in cold start situations, where nearly no user information is available. In this paper, we propose \sys, a POI recommendation system which tackles the challenges of cold start, customizability, contextuality, and explainability by exploiting look-alike groups mined in public POI datasets. \sys reformulates the problem of POI recommendation, as recommending explainable look-alike groups (and their POIs) which are in line with user's interests. \sys frames the task of POI recommendation as an exploratory process where users interact with the system by expressing their favorite POIs, and their interactions impact the way look-alike groups are selected out. Moreover, \sys employs ``mindsets'', which capture actual situation and intent of the user, and enforce the semantics of POI interestingness. In an extensive set of experiments, we show the quality of our approach in recommending relevant look-alike groups and their POIs, in terms of efficiency and effectiveness.
\end{abstract}

\maketitle
\pagestyle{plain}

\section{Introduction}
\label{sec:intro}
There has been a meteoric rise in the use of location-based systems to benefit from services such as bike sharing~\cite{chung2018bike}, localized advertising~\cite{feng2016towards}, urban emergency management~\cite{xu2016crowdsourcing}, and regional health-care~\cite{coviz}. Point-of-Interest (POI) recommendation is one of the most prominent applications of location-based services which benefit both consumers and enterprises. The task of POI recommendation is to recommend a user the POIs (e.g., restaurants, coffee shops, museums) that they may be interested in, but have never visited in a given time window. While POI recommendation in general inherits the large body of work in the community of recommender systems~\cite{bao2015recommendations}, it also carries new constraints and challenges that may not be the case for a traditional recommender, such as spatial distance semantics between POIs and user interactions on maps to select favorite POIs.  An ideal POI recommendation approach should typically capture the following aspects (\textbf{A1} to \textbf{A4}).

\separateshort
\noindent \textbf{A1: Personalization.} First, POI recommendations should be personalized, i.e., the results should be based on user preferences captured in the form of user's historical check-ins and interests.

\separateshort
\noindent \textbf{A2: Socialization.} People trust like-minded users and base their decisions on what people like them appreciated before~\cite{DBLP:conf/chi/DuPSS17}. Hence the POI recommendation should also incorporate social aspects and reflect the preferences of other people similar to the user.

\separateshort
\noindent \textbf{A3: Customization.} Beyond being personalized, the POI recommendation system should also be \textit{exploratory} to incorporate user's interactions with the system, i.e., the customization of recommended POIs. This aspect calls for {\em exploratory scenarios} where the user does not have a precise POI searching criterion in mind, and wants to navigate POIs, and ultimately land on a decision after a few iterations~\cite{paay2016discovering}.

\separateshort
\noindent \textbf{A4: Contextualization.} The POI recommendation should also capture the current situation of the user (aka context). While the literature focuses on time and location as the context, user's actual situation and intent have received less attention~\cite{Anand:2007:CR:1422159.1422167}.

\separateshort
To the best of our knowledge, no POI recommendation approach in the literature addresses all the aforementioned aspects, \textbf{A1} to \textbf{A4}, simultaneously. Personalization (\textbf{A1}) has been the focus of many approaches in the past (\cite{levandoski2012lars,liu2013learning,yin2014lcars}, to name a few), where historical check-ins are exploited to predict which POI the user prefers to visit next, using techniques such as Matrix Factorization and Poisson Factor Model. 
Also, socialization (\textbf{A2}) has been addressed in the literature (\cite{10.1145/2525314.2525357,li2016point}, to name a few), where information encapsulated in location-based social networks (LBSN) are employed to predict user's preferences using link-based methods~\cite{wang2013location}. Below we discuss fundamental challenges for materializing a system which incorporates all the aforementioned aspects.

\separateshort
\noindent \textbf{C1: Cold start.} The problem of cold start arises when a user with a non-existent or limited history of check-ins asks for recommendations. 
A typical recommendation system which relies on historical check-ins and user similarities for personalized and socialized recommendations (\textbf{A1} and \textbf{A2}, respectively) is unable to output results in the presence of cold start.
There are three kinds of users which may cause a cold start: $(i)$ a new user with no history, $(ii)$ a casual user who has not visited the service for a long time, and $(iii)$ private users who do not want their data to be exploited.
Note that the cold start problem refers to the lack of user history and not the lack of user signals, such as user's current time and location. Hence a POI recommendation system can benefit from those available signals to surmount the cold start.

\noindent \textbf{C2: Interpretations of interactions.} Most POI recommendation systems assume the process to be {\em one-shot}, where the user enters the system with a clear unambiguous intent, and the system returns the most interesting POIs which relates to that intent. In practice, this architecture is not realistic. Users need to interact with the system to gradually build their intent. The challenge with multi-shot recommendation systems (\textbf{A3}) is that it is not clear how user interactions with the system should influence the recommendation. 

\begin{figure}[t]
\includegraphics[width=\columnwidth]{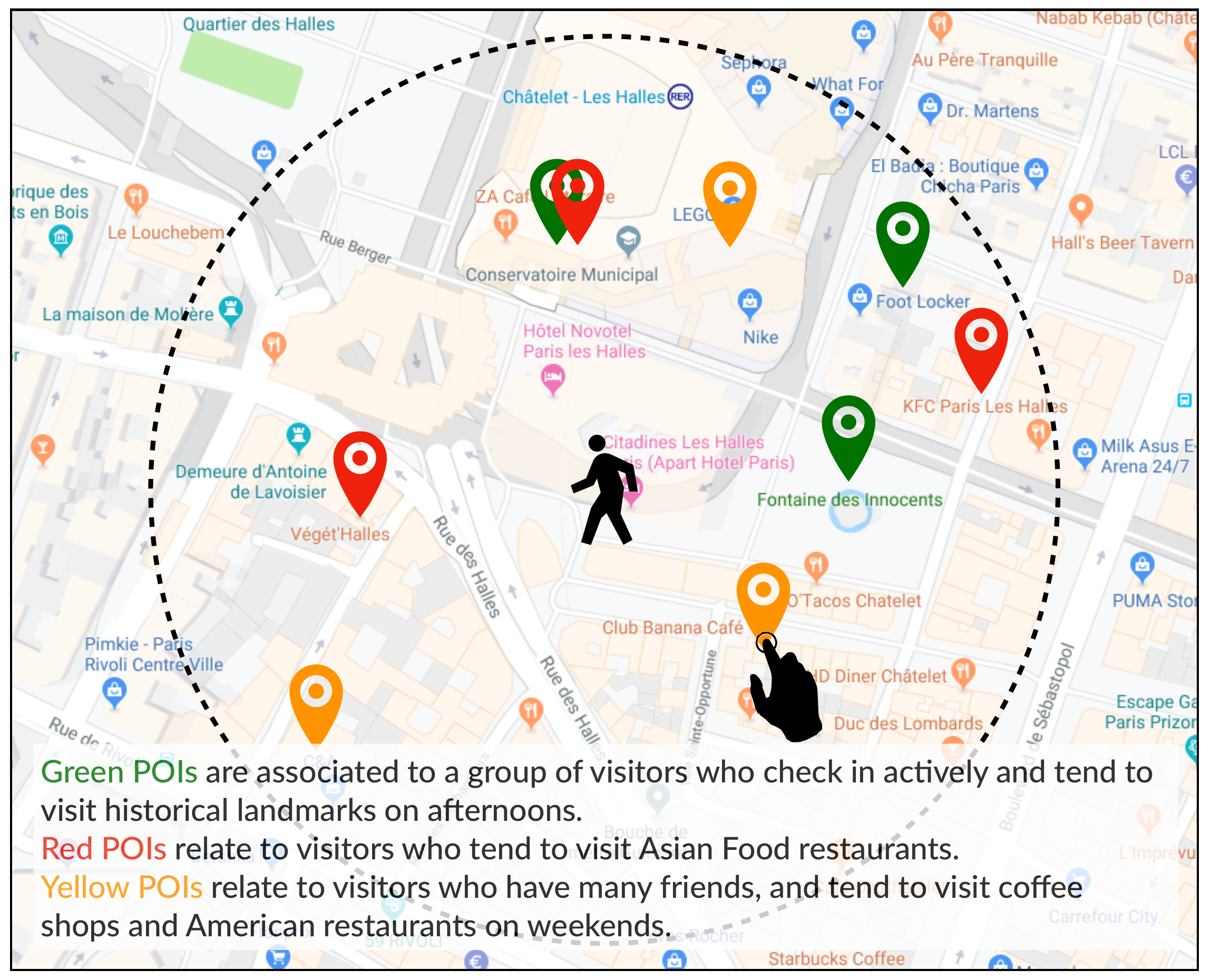}
\caption{\sys in practice.}
 \label{fig:example}
\end{figure}

\separateshort
\noindent \textbf{C3: Context integration.} Another challenge is to integrate the context of the user into the recommendation process (\textbf{A4}). The context is not limited to time and location, but also to the user's mindset at the time of receiving recommendations. For instance, different POIs should be recommended in case the user is hungry, or in case he/she is seeking some personal relaxation time (i.e., me time).

\separateshort
\noindent \textbf{C4: Explainability.} Users may not trust what they get from the recommender, i.e., algorithmic anxiety due to the cold start problem (\textbf{A1}) and session-based interactions with the system (\textbf{A3})~\cite{jhaver2018algorithmic}. Hence it is of crucial importance to let users know \textit{why they receive certain POIs as recommendation results}. 

\separateshort
To collectively address the challenges \textbf{C1} to \textbf{C4}, we propose \sys, an interactive and explainable POI recommendation system based on look-alike groups. The intuition behind \sys is as follows: while it is assumed that no data is available from the user, POI recommendations can be obtained by finding look-alike groups in publicly available POI datasets such as \textsc{Yelp}\footnote{\it https://www.yelp.com/dataset}, \textsc{Foursquare}\footnote{\it https://developer.foursquare.com/docs/places-database/}, 
\textsc{Jiepang}\footnote{\it https://jiepang.com}, 
\textsc{Facebook Places}\footnote{\it https://www.facebook.com/places/},
and \textsc{Gowalla}\footnote{\it https://snap.stanford.edu/data/loc-gowalla.html}. It is shown in the literature that users trust their peers and get inspired by them for decision making~\cite{DBLP:conf/chi/DuPSS17}. The recommended POIs are explainable using their associated groups, e.g., ``the group of photoholics tends to visit Montmartre in the 18th district of Paris'', and ``the group of food lovers tends to visit the restaurant `les Apotres de Pigalles', in the same region.'' The user will then interact with those groups to detect with which group he/she identifies. As a result of this interaction, new groups will be mined, to align with the user's preferences. This iterative process ensures that groups and their POIs reflect user's interests. Note that \sys discovers user's interests and aligns recommendations accordingly, {\em without the need of any historical check-in data from the user}. The following example describes how \sys functions in practice.

\begin{example}
Jane is visiting Paris as a tourist. She is walking in the area of the Pompidou center. After 30 minutes of walking, she gets tired and asks \sys for ``me time'' recommendations, to find POIs in her vicinity (the dashed circle in Figure~\ref{fig:example}) where she can sit and relax. Jane is concerned about privacy and does not share any historical check-in data with the system. \sys outputs three user groups related to Jane's intent, and top-three POIs for each group (Figure~\ref{fig:example}). She looks at group descriptions to see where she can find some doppelg\"angers in group members. Being a social person, Jane shows interest in the yellow group, i.e., visitors who have many friends (i.e., social visitors like her), and tend to visit coffee shops and American restaurants on weekends. This motivates her to reformulate her intent and ask for a place where she can eat. She interacts with the system and asks where people usually eat in the neighborhood. Hence \sys returns another three groups to satisfy Jane's intent (Figure~\ref{fig:example2}). This helps Jane to make up her mind and go to an American-cuisine restaurant.
\end{example}

The example shows that interactive and explainable recommendations help users refine their ambiguous needs and finalize their decision making process on POI selection. The way users can specify their mindsets (e.g., ``me time'', ``I'm hungry'') enables users enforce their context to the recommendation system (i.e., tackling \textbf{C3}) and steer the results towards what they are really interested to receive.


\separateshort
\noindent \textbf{Contributions.} In summary, we propose the following contributions.

\begin{figure}[t]
\includegraphics[width=\columnwidth]{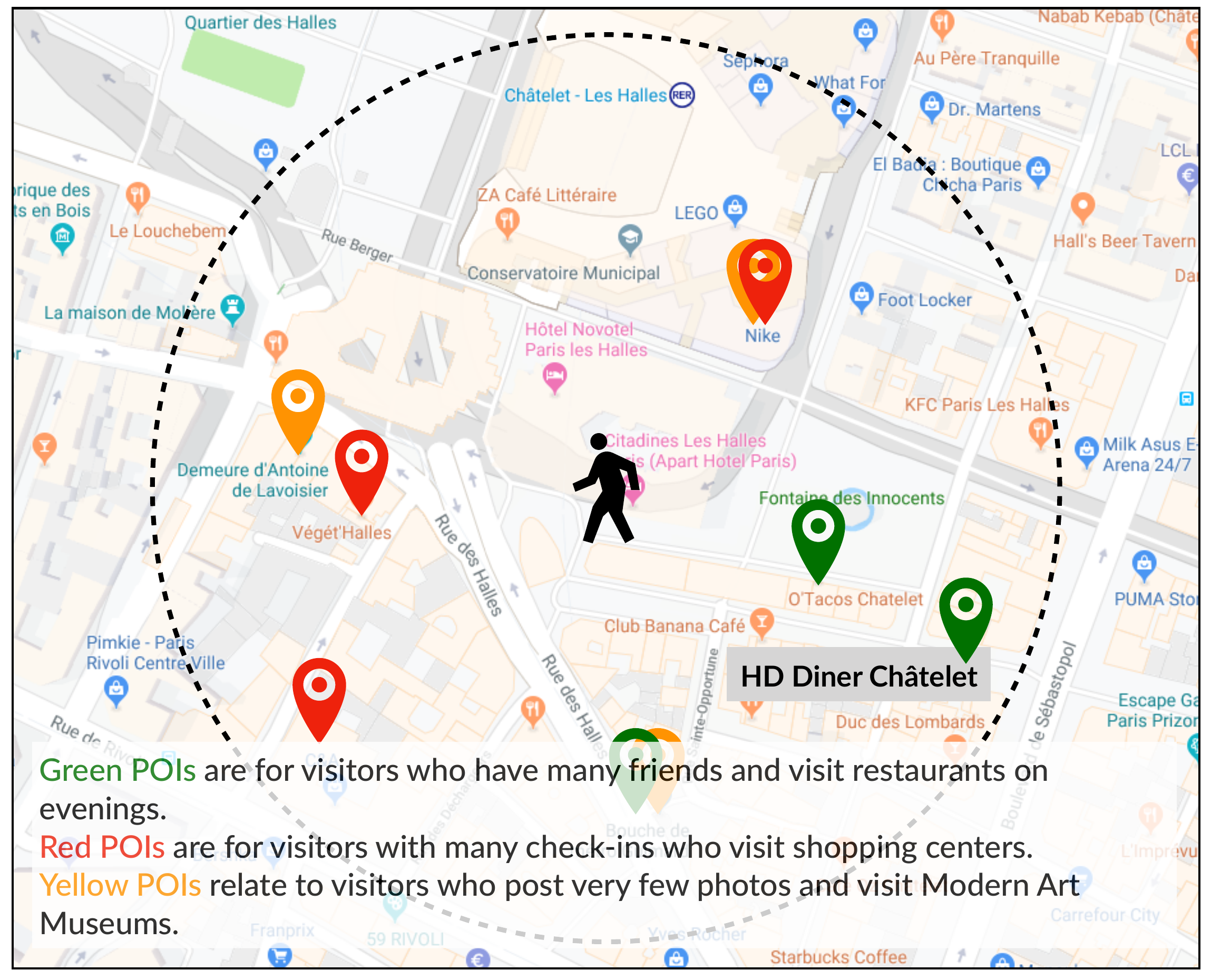}
\caption{Followup iteration in \sys after clicking a POI from the yellow group in Figure~\ref{fig:example}.}
 \label{fig:example2}
\end{figure}

\begin{itemize}[leftmargin=*]
    \item We address the cold start and explainability challenges by employing look-alike user data as a proxy for user preferences. We reformulate the problem of POI recommendation, as recommending explainable groups (and their POIs) which are in line with user's interests.
    \item We consider ``recommendation'' as an exploratory process where the user interacts with the system, and his/her interactions impact the way look-alike groups are selected out. We mention how we build a user portfolio of POIs in order to make recommendations more relevant in further iterations.
    \item We introduce and formalize the notion of ``mindsets'', which captures actual situation and intents of the user. We mention how the interestingness of POIs is maximized using mindsets.
    \item In an extensive set of quantitative and qualitative experiments, we show the efficiency and effectiveness of \sys.
\end{itemize}

\separateshort
The outline of the paper is as follows: in Section~\ref{sec:datamodel}, we provide the data model. In Section~\ref{sec:pb}, we formally define our problem and discuss the challenges. In Section~\ref{sec:algo}, we describe \sys architecture and its underlying algorithm as a solution for interactive and explainable POI recommendation. Section~\ref{sec:exp} presents the detailed experiments. The related work is provided in Section~\ref{sec:rel}. Last, we conclude and discuss future directions in Section~\ref{sec:conc}.

\section{Data Model}
\label{sec:datamodel}
We consider a user $\mu$ asking for POI recommendations. We denote $\mu$'s portfolio as a set $\mathcal{P}_{\mu}$ of POIs that $\mu$ is interested in. Following our assumption of cold start (\textbf{C1}), we consider that initially $\mathcal{P}_{\mu} = \emptyset$. Additionally, we consider a POI dataset $\mathcal{D} = \langle \mathcal{U}, \mathcal{P} \rangle$ with a set of visitors $\mathcal{U}$ and a set of POIs $\mathcal{P}$. 

\separateshort
\noindent \textbf{Visitors.} For a visitor $u \in \mathcal{U}$, the set $u.\mathit{demogs}$ contains tuples of the form $\langle d, v \rangle$ where $d$ is a demographic attribute (e.g., age, gender, number of trips, number of check-ins), and $v \in \mathit{domain}(d)$. Also the set $u.\mathit{checkins}$ contains tuples of the form $\langle p, t \rangle$ which represents that $u$ has visited a POI $p \in \mathcal{P}$ at time $t$.

\separateshort
\noindent \textbf{POIs.} A POI $p \in \mathcal{P}$ is defined as a tuple $p = \langle \mathit{loc}, \mathit{att} \rangle$ where $p.\mathit{loc}$ is itself a tuple $\langle lat, lon \rangle$ (latitude and longitude, respectively) which defines where $p$ is situated geographically. 
The set $p.\mathit{att}$ is a set of tuples of the form $\langle a, v \rangle$ which denote that the POI has the value $v$ for the attribute~$a$, such that $v \in \mathit{domain}(a)$. 

\separateshort
We measure the interestingness of POIs using \textit{POI utility} functions. A more interesting POI has higher chances to be recommended to the user $\mu$. A POI utility function $f: 2^{\mathcal{P}} \mapsto [0, 1]$ returns a value between $0$ and $1$ which reflects the extent of interestingness for one or several POIs~\cite{geng2006interestingness}. Table~\ref{tbl:utility} lists POI utility functions that we employ in this work. We denote the set of all POI utility functions as $\mathcal{F}$. We also define all utility functions as maximization objectives.

\begin{table}[h]
\begin{tabular}{|l|l|}
\hline
\textbf{Utility function} & \textbf{Description} \\ \hline \hline
$\mathit{popularity}(P)$ & \begin{tabular}[c]{@{}l@{}}normalized average number of\\POI check-ins in $P$.\end{tabular} \\ \hline
$\mathit{prestige}(P)$ & \begin{tabular}[c]{@{}l@{}}normalized average rating\\score of POIs in P.\end{tabular} \\ \hline
$\mathit{recency}(P)$ & \begin{tabular}[c]{@{}l@{}}inverse difference between the\\current date and the average\\ insertion date of POIs in $P$.\end{tabular} \\ \hline
$\mathit{coverage}(P)$ & \begin{tabular}[c]{@{}l@{}}the area of the convex hull induced by\\the geographical location of POIs in $P$\\normalized by the size of the city.\end{tabular} \\ \hline
$\mathit{surprisingness}(P)$ & \begin{tabular}[c]{@{}l@{}}normalized Jaccard distance between\\POI categories of $P$ and POI\\categories of the visited POIs\\by the user in $\mathcal{P}_{\mu}$.\end{tabular}  \\ \hline
$\mathit{category}(P,\mathit{cat})$ & \begin{tabular}[c]{@{}l@{}}normalized Jaccard similarity\\between the set $\mathit{cat}$ and the\\POI categories in $P$.\end{tabular} \\ \hline
$\mathit{diversity}(P)$ & \begin{tabular}[c]{@{}l@{}}normalized Jaccard distance\\between sets of POI categories in $P$.\end{tabular} \\ \hline
$\mathit{size}(P)$  & \begin{tabular}[c]{@{}l@{}}normalized average radius\\of POIs in $P$.\end{tabular} \\ \hline
\end{tabular}
\caption{POI utility functions ($P \subseteq \mathcal{P}$).}
\label{tbl:utility}
\vspace{-20pt}
\end{table}

\section{Problem Definition}
\label{sec:pb}
We form our problem definition based on two core assumptions. First, we assume that the user $\mu$ is in an \textit{exploratory setting} and does not necessarily have a clear idea of his/her needs~\cite{paay2016discovering}, and the user is going to sharpen his/her intent in several iterations~\cite{nandi2011guided}. Second, we conjecture that \textit{look-alike user data} is a good proxy to gain user preferences~\cite{behrooztkde}, in the absence of historical check-ins of $\mu$. These two assumptions enable us to define the problem according to the challenges \textbf{C1} to \textbf{C4} presented in Section~\ref{sec:intro}. In this section, first we discuss how the aforementioned assumptions are incorporated in the problem of interactive and explainable POI recommendation. Then we formally define our problem and discuss its hardness.

\subsection{Exploratory settings}
\label{sec:exploratory}
\noindent \textbf{Context.} A core concept in an exploratory POI recommendation setting is ``context'', which is often materialized as the current time and location of the user,
denoted as a tuple $c_{\mu} = \langle \mathit{loc}, \mathit{time} \rangle$ for a given user~$\mu$. For more generality, we often represent $c_{\mu}.\mathit{time}$ as a categorized variable.\footnote{\it We consider the following categories for time: ``morning'' (5AM-11AM), ``afternoon'' (12PM-5PM), ``evening'' (6PM-10PM), and ``night'' (11PM-4AM).} An additional dimension of contextuality is ``mindset'', i.e., the actual situation and intent of the user.

\separateshort
\noindent \textbf{Mindsets.} Mindsets should reflect the way interestingness of POIs are computed based on user's interests. While online services such as \textsc{AroundMe}\footnote{\it http://www.aroundmeapp.com} 
enable users to explore their nearby region by selecting explicit POI categories (e.g., museums, historical landmarks), mindsets capture implicit intents of users (e.g., ``let's learn'') which are more challenging to capture. A mindset $m$ is a tuple $m = \langle \mathit{label}, \mathit{func}() \rangle$, where $\mathit{label}$ provides a short description of the mindset, and $\mathit{func}()$ defines semantics of POI interestingness. 
For instance, in case $m.\mathit{label}=$ {\em ``let's learn''}, $m.\mathit{func}()$ favors museums, libraries and cultural landmarks. We denote the set of all mindsets as $\mathcal{M}$.

\begin{table}[t]
\begin{tabular}{|l|l|}
\hline
\textbf{Mindset label}     & \textbf{Description}                                                                                                                                  \\ \hline \hline
\textbf{$m_1$: I'm new here}  & \begin{tabular}[c]{@{}l@{}}towards touristic POIs about the \\ popular attractions in the city.\end{tabular}                                          \\ \hline
\textbf{$m_2$: surprise me}   & \begin{tabular}[c]{@{}l@{}}towards POIs which haven't been\\ visited before by the user and are\\uncommon (seldom visited)\end{tabular}               \\ \hline
\textbf{$m_3$: let's workout} & \begin{tabular}[c]{@{}l@{}}towards POIs related to physical\\exercises like swimming pools,\\parks, gyms, and mountains\end{tabular}                 \\ \hline
\textbf{$m_4$: me time}       & \begin{tabular}[c]{@{}l@{}}towards POIs related to activities\\to treat oneself and be pursued solo\\to unwind and relax\end{tabular}                                   \\ \hline
\textbf{$m_5$: I'm hungry}    & \begin{tabular}[c]{@{}l@{}}towards getting faster access\\to food-related POIs nearby\end{tabular}                                                                                            \\ \hline
\textbf{$m_6$: let's learn}   & \begin{tabular}[c]{@{}l@{}}towards POIs such as museums,\\libraries and cultural landmarks \end{tabular}                                                                                        \\ \hline
\textbf{$m_7$: hidden gems}   & \begin{tabular}[c]{@{}l@{}}towards intriguing local POIs\\ that are highly rated but\\not necessarily popular\end{tabular} \\ \hline
\end{tabular}
\caption{Mindsets}
\label{tbl:mindsets}
\vspace{-20pt}
\end{table}

\separateshort
\noindent \textbf{Materializing mindsets.} In a field study, with a prototype of a POI recommendation system, we performed an initial concept validation using the Wizard of Oz methodology
to discover how users perceive exploring POIs using different mindsets (see ~\cite{DBLP:conf/chi/ViswanathanOBRG20} for the details of our field study). We recruited a diverse set of $12$ participants for an in-depth analysis of their situation and intent for receiving POI recommendations. During these sessions with our participants, we presented them with $7$ different mindsets, which we had initially hypothesized. Our participants were first asked to choose a mindset, which would in turn lead them to find a POI they would like to visit. In this study, we were able to qualitatively confirm the usefulness of mindsets to find POIs. 
We meticulously paraphrased our participants' interpretation of the mindsets to reach descriptions provided in Table~\ref{tbl:mindsets}. While these mindsets are the result of an in-depth study of POI exploration needs, our model is generic enough to allow new mindsets. We describe the process of mindset creation in Section~\ref{sec:beyond}.

\separateshort
\noindent \textbf{Mindset function.} Given a mindset $m$, the function $m.\mathit{func}()$ is defined as follows.

\begin{equation}\label{eq:func}
m.\mathit{func}(P, \mu) = \frac{\sum_{f_i \in \mathcal{F}} w_{i,\mu}  b_{i,m}  f_i(P)}{\sum_{f_i \in \mathcal{F}} w_{i,\mu} b_{i,m}}
\end{equation}

\begin{table*}[t]
\centering
\begin{tabular}{|l|c|c|c|c|c|c|c|c|}
\hline
                  & \textbf{popularity} & \textbf{prestige} & \textbf{recency} & \textbf{coverage} & \textbf{surprisingness} & \textbf{category} & \textbf{diversity} & \textbf{size}   \\ \hline \hline
\textbf{$m_1$: I'm new here}  & $\boldsymbol{0.25}$     & $\boldsymbol{0.25}$   & $0.10$   & $0.15$   & $0.00$          & $0.00$    & $\boldsymbol{0.25}$    & $0.00$  \\ \hline
\textbf{$m_2$: surprise me}   & $0.25$     & $0.20$    & $0.00$   & $0.00$    & $\boldsymbol{0.30}$         & $0.00$    & $0.15$     & $0.10$  \\ \hline
\textbf{$m_3$: let's workout} & $0.25$      & $0.25$    & $0.00$   & $0.10$    & $0.00$          & $\boldsymbol{0.40}$    & $0.00$     & $0.00$  \\ \hline
\textbf{$m_4$: me time}       & $0.10$      & $0.10$    & $0.00$   & $0.10$    & $0.00$          & $\boldsymbol{0.40}$    & $0.00$     & $0.30$  \\ \hline
\textbf{$m_5$: I'm hungry}    & $0.05$      & $0.20$    & $0.10$   & $0.15$    & $0.00$          & $\boldsymbol{0.40}$    & $0.05$     & $0.05$  \\ \hline
\textbf{$m_6$: let's learn}   & $0.20$      & $0.20$    & $0.00$   & $0.10$    & $0.00$          & $\boldsymbol{0.40}$    & $0.10$     & $0.00$  \\ \hline
\textbf{$m_7$: hidden gems}   & $\boldsymbol{0.30}$      & $\boldsymbol{0.30}$    & $0.15$  & $0.00$    & $0.00$          & $0.00$    & $0.00$     & $0.25$ \\ \hline
\end{tabular}
\caption{Priors in mindsets. The largest values of prior per mindset are shown in bold.}
\label{tbl:bias}
\vspace{-15pt}
\end{table*}

In Equation~\ref{eq:func}, $f_i(P)$ is a utility function (see Table~\ref{tbl:utility}) applied on a set of POIs $P \subset \mathcal{P}$, and $w_{i,\mu}$ and $b_{i,m}$ are the user-specific weight and the prior of the function $f_i$ for the user $\mu$ and the mindset~$m$, respectively. Priors reflect the importance of a utility function for a mindset. In case $b_{i,m} = 0$, it means that the function $f_i$ has no influence on the mindset $m$. On the contrary, in case $b_{i,m} = 1$, it means that the mindset~$m$ is defined only based on $f_i$. User-specific weights, on the other hand, reflect the importance of a utility function for the user. A user may have, for instance, more interest in popularity than coverage. The weights are assumed to shape up when the user interacts with the system. Given the set of all possible user-specific weights~$W$, we initially set $\forall w \in W, w=1.0$. While weights are dynamic and changes per user, priors can be learned offline and stay unchanged at the online execution.

\begin{table*}[]
\begin{tabular}{|l|l|}
\hline
\textbf{Mindset} & \textbf{Categories of interest}                                                                                       \\ \hline
\textbf{$m_3$: let's workout}               & sport fields, park, health and fitness, bowling, tennis court, ice skating, gym                                 \\ \hline
\textbf{$m_4$: me time}                &  outdoor, food, tea room, bar, coffee shop \\ \hline
\textbf{$m_5$: I'm hungry}                & food, restaurant\\ \hline
\textbf{$m_6$: let's learn}                &  museum, art, gallery, library sculpture, bookstore, movie theater, historical landmark, monument  \\ \hline
\end{tabular}
\caption{Categories of interest for mindsets}
\label{tbl:cat}
\vspace{-20pt}
\end{table*}

\separateshort
\noindent \textbf{Utility function priors.}
Table~\ref{tbl:bias} shows a consensus over our field study~~\cite{DBLP:conf/chi/ViswanathanOBRG20} on utility function priors in mindsets. For instance, the mindset $m_2$ (i.e., ``surprise me'') is a combination of the following utility functions in decreasing order of their prior value: {\em surprisingness}, {\em popularity}, {\em prestige}, {\em diversity}, and {\em size}. In case of the {\em category} function, we need to specify POI categories of interest. Table~\ref{tbl:cat} shows these specific categories for each mindset. Note that $m_1$, $m_2$, and $m_7$ are not related to the {\em category} function, hence are excluded from Table~\ref{tbl:cat}. While these categories come from the \textsc{Gowalla} dataset, similar set of categories exists in other POI datasets as well. 


\subsection{Look-alike user data}
\label{sec:lookalike}
There exist various publicly available POI datasets
structured as $\mathcal{D} = \langle \mathcal{U}, \mathcal{P} \rangle$. 
To build look-alike relations in a publicly available POI dataset, we build ``visitor groups'' which aggregate a set of visitors with common demographics and/or POIs~\cite{behrooztkde}. Grouping visitors is in conformance with the ``human mobility behavioral clustering phenomenon'' which mentions that individual visiting locations tend to cluster together~\cite{tobler1970computer}. Visitor groups are obviously virtual and group members do not necessarily know each other. In other words, members of a group are ``location friends'' (who have checked in the same places) and not ``social friends'' (who are socially connected in an LBSN)~\cite{li2016point}. A visitor group is a triple $g = \langle \mathit{members}, \mathit{demogs},$ $\mathit{POIs} \rangle$ where $g.\mathit{members}$ $\subseteq$ $\mathcal{U}$, and ``demogs'' and ``POIs'' contain following expression-based conditions that those members should satisfy:
%


\begin{itemize}[leftmargin=*]
    \item $(i.)$ $\forall u \in g.\mathit{members}, \forall \langle a,v\rangle \in g.\mathit{demogs}, \langle a,v\rangle \in u.\mathit{demogs}$,
    \item $(ii.)$ $\forall u \in g.\mathit{members}, \forall p \in g.\mathit{POIs}, \exists \langle p,t \rangle \in u.\mathit{checkins}$.
\end{itemize}

We also define the relevance between a visitor group $g$ and $\mathcal{P}_{\mu}$ as a function $rel(g,\mathcal{P}_{\mu})$ defined below.

\begin{equation}\label{eq:rel}
  \mathit{rel}(g,\mathcal{P}_{\mu})=
    \begin{cases}
      \frac{|\mathcal{P}_{\mu} \cap g.\mathit{POIs}|}{|\mathcal{P}_{\mu}|} & \mathcal{P}_{\mu} \neq \emptyset\\
      1 & \mathit{otherwise} \\
    \end{cases} 
\end{equation}

Intuitively, a group $g$ is relevant to a user $u$ iff there exist some common POIs between the user $u$ and the group~$g$. In case $\mathcal{P}_{\mu} = \emptyset$ (i.e., cold start), we assume that~$g$ is relevant whatsoever. 

\subsection{Formal problem definition}
\label{sec:problemdef}
We define our problem as follows. Given a user $\mu$ and his/her affiliated context $c_{\mu} = \langle \mathit{loc}, \mathit{time} \rangle$, a mindset $m = \langle \mathit{label}, \mathit{func}\rangle$, a radius~$r$, a relevance threshold $\sigma$, and integers $k$ and $k'$, the problem is to find $k$ groups $G$ and $k'$ POIs for each group in $G$, such that the following conditions are met.

\begin{itemize}[leftmargin=*]
    \item $(i.)$ $\forall g \in G, g$, $rel(g,\mathcal{P}_{\mu}) \geq \sigma$;
    \item $(ii.)$ $\forall g \in G, \forall p \in g.\mathit{POIs}, \mathit{distance}(p.loc,c_{\mu}.loc) \leq r$;
    \item $(iii.)$ $\forall g \in G, \forall p \in g.\mathit{POIs}, \forall u \in g.\mathit{members},$\\ $\forall \langle p,t \rangle \in u.\mathit{checkins}, t = c_{\mu}.\mathit{time}$;
    \item $(iv.)$ $\Sigma_{g \in G} \Big(m.\mathit{func}(g.\mathit{POIs}, \mu)\Big)$ is maximized.
\end{itemize}

The first three conditions ensure that groups are relevant to the user, in vicinity of the user's location, and in the same time category of the user's context. The last condition applies the input mindset to groups, and verifies whether groups' POIs are maximally in line with the mindset.

\separateshort
\noindent \textbf{Problem hardness.} Appendix~\ref{apx:np} shows the theoretical hardness of the aforementioned problem. The problem is also challenging in practice, because the potential number of relevant groups is huge and hence the mindset maximization is not straightforward. Given the complexity of formulating look-alike groups and mindsets in the task of POI recommendation, we consider set-based semantics for our recommendation problem. We plan to consider more complex problems in our immediate future work, e.g., POI sequence recommendation~\cite{10.1145/3274895.3274925,10.1145/3274895.3274958} and travel package recommendation~\cite{DBLP:conf/iui/Amer-YahiaBEOV20,reza2017optimal}.

\section{Algorithm}
\label{sec:algo}
\sys is a session-based system which begins with an ambiguous user's intent for POI recommendation, and ends when he/she is satisfied with the resulting POIs.
Each session consists of a finite sequence of iterations which captures interactions with the user. A new iteration begins by defining a mindset (which may remain the same as the previous iteration), which then results in~$k$ relevant groups and $k'$ POIs for each group. 
At the end of each iteration, the user is free to bookmark some of the recommended POIs as his/her favorites. Hence there are two types of feedback that the user can provide to the system: {\em the mindset} and {\em POI bookmarks}. This multi-shot architecture (i.e., \textbf{A3}) contradicts most traditional single-shot POI recommendation approaches
by incorporating {\em user interactions} in the recommendation~(i.e., addressing \textbf{C2}). Note that the input mindset at each iteration is part of the user signal and not user history, and \sys exploits mindsets to surmount the cold start.\footnote{\it One direction of our future work is to automatically predict mindsets based on other available user signals.}

\separateshort
\noindent \textbf{\sys iterations.} At each iteration, the system returns groups and POIs from a POI dataset (addressing the cold start challenge~\textbf{C1}) based on the functionality of the selected mindset. Algorithm~\ref{algo:iprelo} describes the process in each single iteration:

\separateshort
\noindent \textsf{\textit{Step 1: Neighborhood filtering.}} First, the system finds all the nearby POIs which are at most $r$ kilometers/miles far from the user (line~\ref{ln:getpoi}), where $r$ is a user-defined input parameter. The parameter $r$ enforces the Tobler's First Law of Geography~\cite{tobler1970computer} and ensures that recommendations relate to user's location (i.e., \textbf{C3}). We employ an efficient implementation of \texttt{ST\_DWithin} function in PostGIS to achieve a near real-time retrieval of POIs. We denote the set of nearby POIs as~$P \subseteq \mathcal{P}$. 

\separateshort
\noindent \textsf{\textit{Step 2: Check-in retrieval.}} Given $P$, \sys then retrieves all check-ins of the nearby POIs (line~\ref{ln:getcheckins}). We denote the set of all nearby check-ins as~$H$. A check-in $\langle p,t \rangle$ is in $H$ iff $p \in P$ and $t = c_{\mu}.\mathit{time}$. The second condition conveys that the check-in and the user context should belong to the same time category, e.g., ``morning''. \sys finds which visitors checked in nearby POIs using check-ins.

\separateshort
\noindent \textsf{\textit{Step 3: Group mining.}} Then the system mines groups among checked-in visitors (line~\ref{ln:minegroups}) denoted as $G^*$. Given that $|G^*| \gg k$,  the system finds $k$ groups $G \subset G^*$ (s.t., $|G|=k$) which collectively maximize the mindset function (line~\ref{ln:max}). The group set~$G$ is henceforth aligned to the user's intent expressed in the mindset.

\separateshort
\noindent \textsf{\textit{Step 4: POI selection.}} Finally, \sys picks top-$k'$ POIs for each group which are visited by most of the group members (line~\ref{ln:poisinside}). At this stage, the user observes $k$ groups and $k'$ POIs for each, and may decide to bookmark POIs for enriching her portfolio.

\begin{algorithm}[h]
\DontPrintSemicolon
\KwIn{Visitors $\mathcal{U}$ and POIs $\mathcal{P}$, user context $c_{\mu} = \langle \mathit{loc}, \mathit{time} \rangle$, mindset $m$, radius $r$, number of groups $k$, number of POIs per group $k'$}
\KwOut{Groups $G$ and their POIs $P_G$}
$P \gets \mathit{nearby\_POIs}(\mathcal{P}, c_{\mu}.\mathit{loc}, \mathit{r})$\label{ln:getpoi}\;
$H \gets \mathit{checkins\_of}(P, c_{\mu}.\mathit{time})$\label{ln:getcheckins}\;
$G^* \gets \mathit{mine\_groups}(\mathcal{U}, H)$\label{ln:minegroups}\;
$G \gets \mathit{maximize}(u, G^*, k, m)$\label{ln:max}\;
\lFor{each group $g \in G$}{
$P_G.\mathit{append}(\mathit{top\_POIs}(g, k'))$\label{ln:poisinside}}
\Return{$G$, $P_G$}\; 
\caption{\sys Algorithm}
\label{algo:iprelo}
\end{algorithm}

In the following, we discuss how groups are mined (i.e., the $\mathit{mine\_groups}()$ function in line~\ref{ln:minegroups}), how mindset functions are maximized (i.e., the $\mathit{maximize}()$ function in line~\ref{ln:max}), and how bookmarking POIs improves the recommendations in further iterations.
 
\subsection{Mining groups}
To address the challenge of explainability (\textbf{C4}), we aim to find {\em describable groups} which identify a set of visitors checking in a set of POIs. For this aim, we employ Frequent Itemset Mining (FIM) technique, where each group is a frequent itemset, and items are common demographic attributes and POIs of the group members. While groups can be discovered in myriad ways~\cite{10.1145/3347146.3359346,10.1145/2525314.2525318,behrooztkde}, we choose FIM to obtain describable groups with overlaps, so that visitors can be a member of more than one group and be described in different ways. For each visitor $u \in \mathcal{U}$, we build a transaction embedding $trans(u)$ which contains all demographic attributes and visited POIs of $u$. Appendix
~\ref{apx:embeddings} describes the process of building and enriching the embeddings. The set~$\tau$ contains the embeddings of all visitors, i.e., $\tau = \{\mathit{trans}(u) | u \in \mathcal{U}\}$. Given an arbitrarily group of visitors $g = \langle \mathit{members}, \mathit{demogs},$ $\mathit{POIs} \rangle$, we define the group support $\mathit{supp}(g)$ as a measure of $g$'s significance (Equation~\ref{eq:supp}).

\begin{equation}\label{eq:supp}
\begin{split}
\mathit{supp}(g) = |\{u \in \mathcal{U} | & \mathit{trans}(u) \in \tau, \\ & g.\mathit{demogs} \cup g.\mathit{POIs} \subseteq \mathit{trans}(u)\}|
\end{split}
\end{equation}

The group $g$ is a frequent itemset (and hence a valid group) iff $\mathit{supp}(g)$ is larger than a predefined support threshold. While the support threshold is often a user-defined parameter, we materialize all groups with support larger than~$1$, hence no threshold needs to be specified.

Given the set $\tau$, we employ the \textsc{Extract} algorithm~\cite{feddaoui2016extract} to mine groups. The input of the algorithm is of type ARFF (Attribute-Relation File Format) which contains the set of all enriched embeddings. The frequent itemset mining algorithm is known to be inefficient for large number of transactions, as its execution time grows exponentially with the number of transactions. This is why that the algorithm is often considered as an offline step preceding an online investigation on groups. In \sys, we are able to perform the mining process on-the-fly, thanks to the neighborhood filters preceding the algorithm. The algorithm mines groups only for visitors with check-ins in the vicinity of the user. Hence the size of the visitor set is drastically reduced compared to~$|\mathcal{U}|$.

\subsection{Maximizing mindsets}
\label{sec:maxmindset}
Not all groups are equally interesting to the user. We need to pick $k$ groups out of all mined groups which are in line with the mindset requested by the user. This will tackle the challenge of context integration (\textbf{C3}). Each mindset is associated to a function which is a set of utility functions combined in a linear fashion with user-specific weights and priors (Equation~\ref{eq:func}). The mindset function admits as input a set of POIs, and returns a value in the range $[0,1]$. Given a mindset $m$ and a group $g$, we measure the utility of~$g$ regarding $m$'s functionality, as follows.

\begin{equation}
\mathit{group\_utility}(g) = m.\mathit{func}(g.\mathit{POIs})
\end{equation}

Given the space of all group utility values, the problem is to find~$k$ groups with the largest values of group utility. As each mindset function is constructed as a combination of several utility functions, maximizing mindset functions is a multi-objective optimization problem in nature. However, we employ a scalarization approach in \sys using user-specific weights and priors to reduce the complexity of the problem to single-objective optimization. One direction of future work is to directly employ multi-objective optimization approaches to obtain skyline groups. 

\separateshort
In \sys, we employ a greedy-style hill climbing algorithm to maximize mindset functions (Algorithm~\ref{algo:max}):

\separateshort
\textsf{\textit{Step 1: Pruning.}} The algorithm starts by removing all irrelevant groups to the user~$\mu$ (line~\ref{ln:irrelevant}). This ensures that groups are inline with user's preferences provided in previous interactions. While \sys focuses on the cold start problem and employs no historical data for building recommendations, it can however benefit from user's interactions with the system to align recommendations with his/her preferences. We enable users to describe their preferences by ``bookmarking'' recommended POIs. These bookmarked POIs will enrich the user's portfolio, which later contributes to identifying irrelevant groups to user's preferences. Once the user~$\mu$ clicks on a POI $p$, it will be added to $\mathcal{P}_{\mu}$. Then the relevance of a group $g$ regarding $\mathcal{P}_{\mu}$, i.e., $\mathit{rel}(g,\mathcal{P}_{\mu})$ is computed as in Equation~\ref{eq:rel}.

\separateshort
\noindent \textsf{\textit{Step 2: Mindset function maximization.}} After pruning irrelevant groups, the algorithm iterates over the space of groups to maximize the mindset function. At each step, the algorithm introduces a new group to the set of $k$ groups, and verifies if the value of the mindset function increases. In case of improvement, the new group will become a member of the $k$ groups, and another group will be selected for a substitution. The improvement loop breaks when a time limit~$tl$ exceeds. While the time limit $tl$ should normally be a user-defined input parameter, we inspire from ``progressive analytics''~\cite{DBLP:journals/dagstuhl-reports/FeketeF0S18} and fix the time to ``continuity preserving latency'', i.e., $tl=100ms$, which is the the limit for humans to have an instantaneousness experience of the interactive process.


\begin{algorithm}[t]
\DontPrintSemicolon
\KwIn{User $\mu$, mined groups $G^*$, mindset $m$, relevance threshold $\sigma$, number of returned groups $k$, time limit $\mathit{tl}$}
\KwOut{Groups $G$ s.t., $|G|=k$}
\For{$g \in G^*$}
{
\lIf{$\mathit{rel}(g,\mathcal{P}_u) < \sigma$}
{
$G^* \gets G^* \setminus \{g\}$\label{ln:irrelevant}
}
}
$\mathit{sort}(G^*,\mathit{supp},\mathit{descending})$\label{ln:sort}\;
$G \gets \mathit{pick}(G^*, k)$\label{ln:pick1}\;
$g_{out} \gets \mathit{pick}(G^*, 1)$\label{ln:pick2}\;
\While{$tl$ is not exceeded}
{
\For{$g_{in} \in G$}
{
$G' \gets G \cup g_{out} \setminus g_{in}$\;
\If{$\sum_{g' \in G'}m.\mathit{func}(g') > \sum_{g \in G}m.\mathit{func}(g)$}
{
$G \gets G'$\;
$\mathit{break}$\;
}
}
$g_{out} \gets \mathit{pick}(G^*, 1)$\;
}
\Return{$G$}\; 
\caption{$\mathit{maximize}$ function}
\label{algo:max}
\end{algorithm}

The function $\mathit{pick}()$ (lines \ref{ln:pick1} and \ref{ln:pick2}) may employ different semantics to enforce the ``scanning order'' in the space of groups $G^*$. The semantics is usually designed in a way to boost the optimization process by moving faster towards the optimized value~\cite{DBLP:journals/vldb/Omidvar-Tehrani19}. In our case, mindsets are combinations of different optimization objectives. We employ the ``support'' measure (Equation~\ref{eq:supp}) to enforce order, hence larger groups have higher chances to be selected. Larger groups contain more visitors and provide richer insights~\cite{behrooztkde}. Hence the function $\mathit{pick}(G^*,k)$ returns top-$k$ unseen groups in $G^*$ with the largest support.

\subsection{Improving recommendations}
The recommendations should become more personalized (addressing \textbf{C1}) in further interactions of the user with the system (addressing \textbf{C2}). To achieve this, we update weights in mindset functions according to user's interactions. Given a user $\mu$ and a mindset $m$, the weight of a function $f_i$ in $m.\mathit{func}()$ is calculated as follows.

\begin{equation}
  w_{i,\mu} =
    \begin{cases}
      f_i(\mathcal{P}_{\mu}) & \mathcal{P}_{\mu} \neq \emptyset \\
      1 & \mathit{otherwise} \\
    \end{cases} 
\end{equation}

In case $\mu$ has already performed some interactions with the system (i.e., $\mathcal{P}_{\mu} \neq \emptyset$), the value of $w_{i,\mu}$ is set relative to the orientation of $\mu$'s previous choices towards $f_i$. It reflects the importance of a utility measure for the user~$\mu$. In case no interaction is available, the weight is set to $1.0$. For instance, in case $f_i = \mathit{coverage}$, and $\mathcal{P}_{\mu}$ contains POIs only from one single neighborhood of the city, $w_{i,\mu}$ will probably receive a value close to zero. Intuitively, it shows that there is less importance in $\mu$'s subjective preference for 
coverage.

It is important to notice that $\mathcal{P}_{\mu}$ impacts \sys in two different levels of granularity, group and mindset levels. In the group level, it enables the system to have an early pruning on groups with no POI overlap with $\mathcal{P}_{\mu}$ (line~\ref{ln:irrelevant} of Algorithm~\ref{algo:max}). In the mindset level, it aligns user-specific weights of utility functions to user preferences in $\mathcal{P}_{\mu}$.

\vspace{-5pt}
\subsection{Beyond predefined mindsets}
\label{sec:beyond}
While we propose a set of mindsets in Table~\ref{tbl:mindsets}, the functionality of \sys is not dependent on a predefined set of mindsets and can be easily extended to new ones. A new mindset $m^*$ is identified by its priors $b_{i,m^*}$. The new mindset can be created either from the ground up, or by combining some predefined mindsets. In the former case, $m^*$ is derived from a set of interesting POIs $P$ provided by the user, where $b_{i,m^*} = f_i(P)$. 
In the latter case, $m^*$ is constructed as an aggregation of a subset of predefined mindsets $M \subseteq \mathcal{M}$, where $b_{i,m^*}=\mathit{avg}_{m \in M}(b_{i,m})$.

\section{Experiments}
\label{sec:exp}
In this work, we use \textsc{Gowalla} dataset~\cite{liu2014exploiting}, collected from a popular LBSN with $36,001,959$ check-ins of $319,063$ visitors over $2,844,076$ POIs.
\textsc{Gowalla}'s check-in matrix density is $2.9 \times 10^{-5}$. POIs are grouped in 7 different categories, i.e., community, entertainment, food, nightlife, outdoors, shopping, and travel. 
\textsc{Gowalla} is among the few POI datasets that provide attributes both for visitors and POIs. Hence we can form groups containing both demographic attributes and POIs. This increases the explainability of groups, and enables users find out in which group they identify. Visitor attributes are illustrated in Table~\ref{tbl:discretization}. Also POIs are described using the following attributes: insertion date, location, total number of check-ins, radius, 
and categories. In this section, we present a comprehensive set of experiments to validate the efficiency and effectiveness of \sys. 

\separateshort
\noindent \textbf{Experimental methods.} Unlike one-shot recommendation algorithms, an exploratory system such as \sys incorporates end-users in the loop, hence the need for an extensive validation of both algorithmic and user-centric aspects. Section~\ref{sec:simstudy} evaluates the overall algorithmic behavior of \sys by simulating interaction sessions and reporting the Hit Ratio measure. Section~\ref{sec:userstudy} evaluates the human-oriented aspects of \sys using an extensive user study in Amazon Mechanical Turk. A complementary viewpoint on the performance aspects of \sys is also provided in Appendix~\ref{apx:perf}.

\separateshort
\noindent \textbf{Baselines.} Note that the functionality of \sys is fundamentally different from the most POI recommendation approaches in the literature. First, most methods ranging from Collaborative Filtering to deep learning approaches require the checkins matrix as input, which is assumed to be nonexistent in the case of \sys. Second, they return a ranked list, while set-based semantics is employed in \sys. For a fairer comparison, we employ the two following preliminary baselines: $(i)$ \textsf{Diversity} which returns a diversified set of POIs in user's vicinity, and $(ii)$ \textsf{Popularity} which returns the most popular POIs.


\subsection{Simulation Study}
\label{sec:simstudy}
Our goal in this part of our experiment is to examine the overall behavior of \sys in tackling the cold-start problem, and providing customizability, contextuality, and explainability in POI recommendation. Preferences of real people will be verified later in the user study (Section~\ref{sec:userstudy}). To remove the influence of human decisions from the exploratory process of \sys, we {\em simulate interactions} in the \textsc{Gowalla} dataset. We then report the Hit Ratio (HR) for each simulated session. 

We simulate 100 different sessions and report HR as the average over all the sessions. In each session, first we randomly pick a user~$\mu$ from \textsc{Gowalla}. From the set $\mu.\mathit{checkins}$, we randomly pick a check-in $\langle p, t \rangle$, and set $\mu$'s context as $c_{\mu} = \langle p.\mathit{loc}, t \rangle$. We build a set $\zeta_{\mu} \subseteq \mu.\mathit{checkins}$ as follows.

\begin{equation}
\begin{split}
\zeta_{\mu} = \{ \langle p, t \rangle \in \mu.\mathit{checkins} |& \mathit{distance}(p.loc,c_{\mu}.loc) \in (0,r] \\ & \wedge |t - c_{\mu}.\mathit{time}| \leq 48\}
\end{split}
\end{equation}

The set $\zeta_{\mu}$ contains nearby POIs (constrained using the radius~$r$) that $\mu$ visited in the next two days following his/her current context time $c_{\mu}.\mathit{time}$. In this experiment, we fix $r=0.5km$. \sys contributes to HR if its output overlaps with $\zeta_{\mu}$. To simulate the cold-start, the {\em whole} set $\mu.\mathit{checkins}$ is masked off as testing set.

Each session contains $N$ consecutive iterations. Each iteration is fired by simulating the action of picking a mindset~$m$. \sys will then generate groups and their POIs using $m$ and $c_{\mu}$. The iteration ends by simulating the action of picking a group of interest $g^*$ and a POI of interest $p^*$ associated with $g^*$. The POI $p^*$ will be added to $\mathcal{P}_{\mu}$. We propose two baselines for group selection: $(i)$ pick a group at random; $(ii)$ pick a group $g^*$ where $\mathit{Cosine}(g^*.\mathit{demogs}, \mu.\mathit{demogs})$ is maximal. We call the latter the {\em optimal group} strategy. Moreover, we propose two baselines for mindset selection: $(i)$ pick one of the mindsets in Table~\ref{tbl:mindsets} at random; $(ii)$ pick the mindset~$m$ with maximal value of $m.\mathit{func}()$ for the POIs of $g^*$ in the previous iteration. We call the latter the {\em optimal mindset} strategy. For a more realistic simulation of the mindset selection process, we consider the fact that users do not always switch to a new mindset in each iteration. Hence we also consider a parameter $\theta$ which is the probability that the mindset in the next iteration will stay unchanged. For instance, in case $\theta = 0.8$, it is highly probable that the mindset does not change between consecutive iterations. By default, we employ random strategy for both group and mindset selection, and $\theta = 0.5$.

We report HR at two different levels of granularity, iteration level and session level. At the {\em iteration level}, we denote the measure as $\mathit{HR}_{I}@N$, and compute it as the average HR obtained at each iteration in a session. 

\begin{equation}
\label{eq:iterationhr}
\mathit{HR}_{I}@N = \frac{\sum_{i=1}^{S}(\frac{\sum_{j=1}^{N}\mathbbm{1}(i,j,\mu)}{N})}{S}
\end{equation}

\separateshort
In Equation~\ref{eq:iterationhr}, $N$ is the number of iterations, $S$ is the number of sessions (in our case, $S=100$), and $\mathbbm{1}(i,j,\mu)$ is an iteration-based hit indicator function which returns ``1'' if there is a hit (POI in common with $\zeta_{\mu}$) at iteration $j$ of the session $i$. At the {\em session level}, we denote the measure as $\mathit{HR}_{S}@N$, and compute it as the average HR for the whole session, as the result of all interactions in $N$ iterations.

\begin{equation}
\label{eq:sessionhr}
\mathit{HR}_{S}@N = \frac{\sum_{i=1}^{S}(\mathbbm{1}_{j=1}^{N}(i,j,\mu))}{S}
\end{equation}

In Equation~\ref{eq:sessionhr}, $\mathbbm{1}(i,j,\mu)$ is a session-based hit indicator function which returns ``1'' if there is at least one hit in the iterations $j$ of the whole session $i$. Obviously session level HR subsumes the one at the iteration level.

Figure~\ref{fig:simulation} illustrates HR values using different strategies of group selection, mindset selection, and mindset change. We measure HR by varying $N$
from $2$ to $50$. The figures on the left report iteration level HR, and the ones on the right report session level~HR. To reduce clutter, we don't show the baseline results on the figures, where \textsf{Popularity} has a constant value of $0.005$ for $\mathit{HR}_{S}@N$, and \textsf{Diversity} grows from zero to $0.001$.

\begin{figure}[t]
\includegraphics[width=\columnwidth]{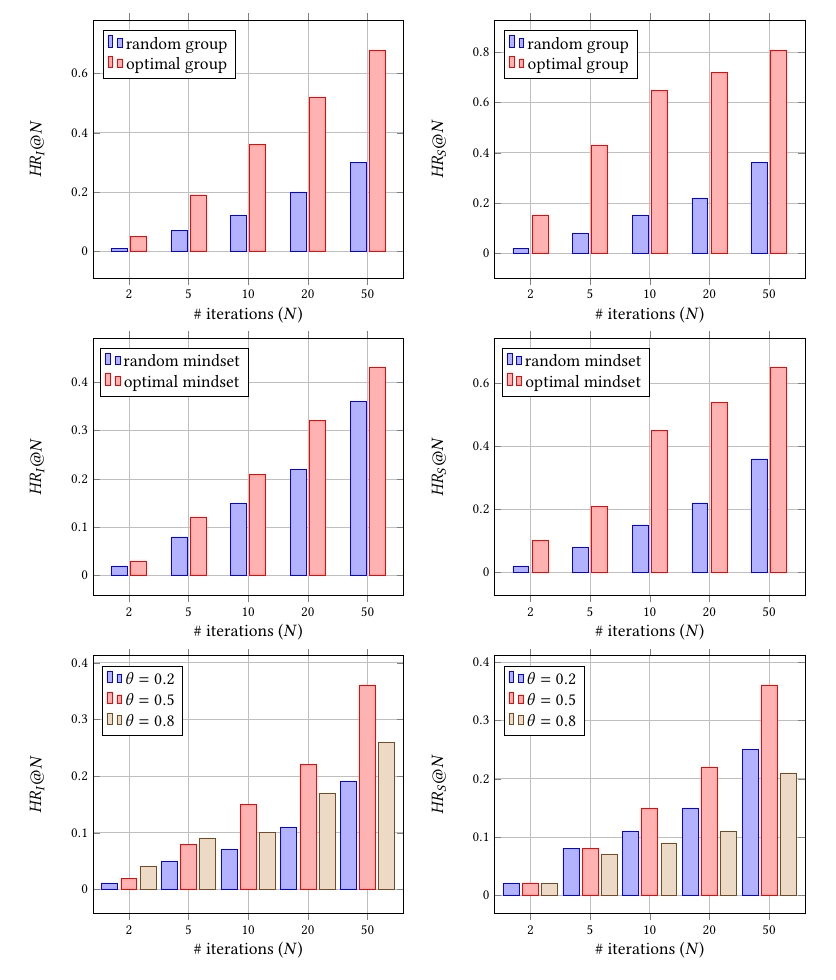}
\caption{Execution time of \sys.}
\label{fig:performance}
\end{figure}

\separateshort
\noindent \textbf{Group selection strategy.} We observe that optimal groups increase HR by $22\%$ and $38.8\%$ on average, in iteration and session levels, respectively (top row in Figure~\ref{fig:simulation}). In session level, the HR grows to values larger than $50\%$ after only $10$ iterations, and it reaches to $82\%$ at $50$ iterations. This confirms our assumption that look-alike groups function as a good proxy to gain user preferences.

\separateshort
\noindent \textbf{Mindset selection strategy.} We observe that optimal mindsets increase HR by $22.4\%$ in the session level (middle row in Figure~\ref{fig:simulation}). However the increase in the iteration level is not salient. This shows that in contrary to group selection strategies, mindset selection has a long-term influence on the overall session. In the session level, the HR grows to values close to $50\%$ after $10$ iterations, and it reaches to $65\%$ at $50$ iterations. 

\separateshort
\noindent \textbf{Changing mindsets.} We also study the effect of changing mindsets between consecutive iterations. We consider three different values for the $\theta$ parameter: $0.2$, $0.5$, and $0.8$. The larger values of~$\theta$ increase the probability that mindsets stay unchanged between consecutive iterations. Bottom row in Figure~\ref{fig:simulation} shows the results. We observe that both extremely low and high values of~$\theta$ lead lower values of HR compared to the other case. The effect is even amplified in more iterations. We conclude that the optimal strategy for changing mindsets is a mixed strategy, i.e., the mindsets do not necessarily change at each iteration, but there is an equal chance that the previous mindset is employed again.

\subsection{User Study}
\label{sec:userstudy}
In this section, we discuss human-oriented aspects of \sys. We perform an extensive between-subject user study in Amazon Mechanical Turk\footnote{\it https://www.mturk.com} (AMT) to measure the effectiveness of employing look-alike groups and mindsets in \sys. Appendix~\ref{apx:userstudy} provides details of this study. 

\subsubsection{Comparative study} 
In the first part of our study we compare \sys with the following baselines: \textsf{Diversity}, \textsf{Popularity}, and \textsf{Google Map Explore}, i.e., the most popular map navigation application\footnote{\it https://themanifest.com/mobile-apps/popularity-google-maps-trends-navigation-apps-2018}.
Each baseline returns $k \times k'$ POIs. The participant observes the results of \sys and each baseline side-by-side, and should decide which set of POI recommendations is more useful. To remove bias, we don't reveal the identity of the method, and shuffle the order of their appearance for each participant. We found out that $98\%$, $87\%$, and $59\%$ of the participants prefer \sys over \textsf{Diversity}, \textsf{Popularity} and \textsf{Google Map Explore}, respectively. While the identity of the methods were concealed, the participants have most probably used \textsf{Google Map Explore} before, but not \sys. This a-priori familiarity could be an influencing factor in this study, which led many participants to vote on their ``comfort zone'' technological method.

\begin{figure}[t]
\includegraphics[width=\columnwidth]{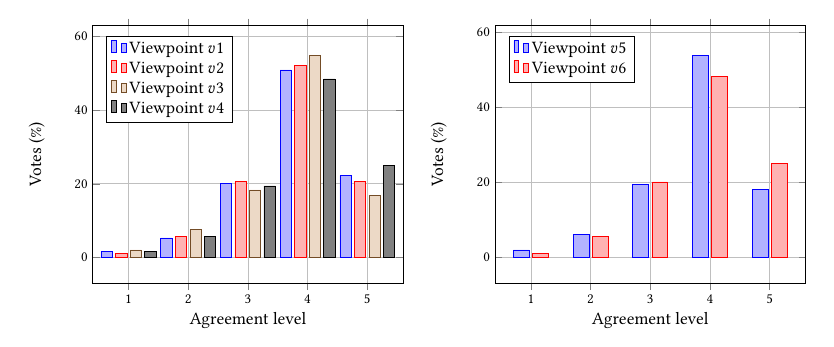}
\caption{\sys user study.}
\label{fig:userstudy} 
\vspace{-15pt}
\end{figure}

\subsubsection{Independent study} 
In the next part of the user study, we delve into more details of evaluating look-alike groups and mindsets in \sys. For each participant in our study, we describe a context (time and location, e.g., ``an evening in Paris'') and a mindset. Then we show the results of \sys ($k$ groups and their $k'$ POIs) and ask the participant's opinion about two different aspects of our system: group-based and interaction-based. The former relates to the usefulness of look-alike groups in POI recommendation. The latter is also about the usefulness of a multi-shot recommender where participants can enrich their intent in several iterations. For both aspects, we pose a set of viewpoints, where the participant expresses his/her agreement with each viewpoint in a Likert scale with the following scores: $1$ (totally disagreed), $2$ (disagreed), $3$ (neutral), $4$ (agreed), and $5$ (totally agreed).

\separateshort
\noindent \textbf{Usefulness of look-alike groups.} Figure~\ref{fig:userstudy}-left reports the group-based aspects. The first two viewpoints evaluate the overall usefulness of look-alike groups in POI recommendation: $(v1)$ the group-based results are relevant to the selected mindset, $(v2)$ the groups help the participant understand why he/she receives the POIs as recommendation results. For both viewpoints, we observe that the supremacy of the agreement vote is statistically significant. This validates our hypotheses on the relevance of group-based recommendations to the mindset (i.e., viewpoint $v1$) and the understandability of recommended POIs using groups (i.e., viewpoint $v2$). We conducted one-way repeated ANOVA for the both results, and obtained the~$F$ statistics of $311.97$ and $731.30$ with the significance level of $0.05$, for $v1$ and $v2$ respectively. 

\separateshort
\noindent \textbf{Usefulness of group descriptions.} The viewpoints $v3$ and $v4$ evaluate the usefulness of group descriptions in the recommendation process: $(v3)$ reading group descriptions helps the participant in deciding about recommended POIs, $(v4)$ the group descriptions are relevant to the recommended POIs. In $v3$, we observe that the majority of votes belongs to the agreement side, i.e., $50.81\%$ agreed, and $22.16\%$ strongly agreed. Regarding $v4$, while $52.11\%$ of participants voted ``agreed'', there is a tie between ``neutral'' and ``totally agreed'' at $20.54\%$. However, the agreement vote is statistically significant in both results, with $F$ statistics of $700.17$ and $499.35$, for $v3$ and $v4$ respectively (using one-way repeated ANOVA with the significance level of $0.05)$. We also asked our participants which part of the group description is more helpful in enabling the decision making about recommended POIs. 
We found that $55.11\%$ of our participants vote ``POI categories'' as the most helpful content in group descriptions, followed by $34.14\%$ for ``visitor common demographics'', and $10.74\%$ for ``time-based information''.

\separateshort
\noindent \textbf{Usefulness of interactions.} Figure~\ref{fig:userstudy}-right reports the interaction-based aspects. Each participant is asked to navigate from one iteration to another, either by choosing a new mindset or preserving the current one. We consider two viewpoints for the interactions: $(v5)$ The step-by-step interaction is useful for the user to follow his/her interests, $(v6)$ the iterative process helps the user to reach a decision on which place he/she would like to visit. We observe $72.15\%$ and $48.43\%$ of agreement votes, and less than $10\%$ of disagreement votes, for $v5$ and $v6$ respectively. This shows the indispensable need of an iterative process for POI recommendation.

\begin{table*}[t]
\begin{tabular}{|l|l|}
\hline
\textbf{Mindset} & \textbf{Comments} \\ \hline \hline
\textbf{$m_1$: I'm new here} & \begin{tabular}[c]{@{}l@{}} {\em ``I find this helpful, because most people have similar interests on tourism.'' ``The group}\\ {\em description allows you to follow a group that conforms to your tastes.''  ``I love to travel, so} \\ {\em these recommendations would be great, considering that I've never been to that location before.''} \\ {\em ``When visiting a new country, I like to take recommendations from others on where to go} \\ {\em and what to see.'' ``I appreciate the way the system categorized the locations because it allows} \\ {\em me to understand what type of curiosity it might fulfill.''}\end{tabular} \\ \hline
\textbf{$m_2$: surprise me} & \begin{tabular}[c]{@{}l@{}}{\em ``Since I was interested in visiting places unlike what I had seen before, I believe it makes}\\ {\em the most sense to base my selections off of the proposed groups.'' ```Surprise me' would be} \\ {\em helpful to just pick a place and go.'' ``The proposed groups seemed to offer out-of-the-box} \\ {\em places to visit, which went with my interest.''}\end{tabular} \\ \hline
\textbf{$m_3$: let's workout} & \begin{tabular}[c]{@{}l@{}} {\em ``I would want to go where other people found the most invigorating exercise that burned the} \\ {\em most calories.'' ``These are diverse adventures which fit my energetic personality.''} \\ {\em ``I appreciate that the system is showing me where I can exercise and do physical activities}\\ {\em like biking.'' ``I like physical activities but also enjoy eating a healthy meal or going to a park}\\ {\em to walk and get some exercise. So this was helpful to me that the system showed me places} \\ {\em where I can do that based on group visits/recommendations.''}\end{tabular}      \\ \hline
\textbf{$m_4$: me time}  & \begin{tabular}[c]{@{}l@{}} {\em ``The groupings would help me avoid overly busy places if I want to relax.'' ``I needed some} \\ {\em `me time' and felt that the Village Saint Paul was the best for me, with cobblestone streets,} \\ {\em a nice museum and a quaint coffee shop.'' ``A good coffeehouse or caf\'e could be helpful,} \\ {\em especially at a specific time of day, or if it is a place favored by locals.''}\end{tabular} \\ \hline
\textbf{$m_5$: I'm hungry} & \begin{tabular}[c]{@{}l@{}} {\em ``It is helpful, as you can see at a glance what restaurants and coffee shops are around you.''} \\ {\em ``I like that the places were in close proximity for meal choices.'' ``With regards to food} \\ {\em preferences, the common categories are good descriptors, because they show a little bit} \\ {\em about what kind of food experience (mainstream or local) that a person would want.''} \\ \end{tabular}  \\ \hline
\textbf{$m_6$: let's learn}  & \begin{tabular}[c]{@{}l@{}} {\em ``It gave me a chance to catch up on history.'' ``It allows you to focus on more museum-oriented} \\ {\em travelers.'' ``Since I was looking for learning opportunities, I believe it was the most useful to look} \\ {\em at the common categories of places that group members visit. ``These places are exactly what I} \\ {\em would visit Paris for.'' ``I was able to pick out places where I could learn more about the history} \\ {\em of Paris.''}\end{tabular} \\ \hline
\textbf{$m_7$: hidden gems}  & \begin{tabular}[c]{@{}l@{}} {\em ``I found this helpful because the places are not very popular, but very intriguing and interesting.''} \\ {\em ``I love the recommendations which are not common place'' ``I believe that the demographics were} \\ {\em the most useful in this scenario. To me, the individuals with many friends would be able to `get} \\ {\em the scoop' on hidden gems in the area.'' ``This helped to acknowledge smaller but good places} \\ {\em to visit that may have otherwise been missed.'' ``If I was looking for hidden gems, I would prefer to} \\ {\em go somewhere not frequented by other people at all.''}\end{tabular} \\ \hline
\end{tabular}
\caption{Participants' opinions about resulting groups and POIs for each mindset}
\label{tbl:comments}
\end{table*}

\separateshort
\noindent \textbf{Usefulness of mindsets.} Beyond agreement scores on viewpoints, we asked our participants' opinion about resulting groups and POIs {\em per mindset}. They could describe their perspective in form of a free text. An extract of these comments is listed in Table~\ref{tbl:comments}. For instance, they mentioned that the mindset $m_1$ (``I'm new here'') is useful for {\em tourism}, $m_2$ (``surprise me'') enables access to {\em out-of-the-box} places, and $m_7$ (``hidden gems'') helps to spot places which are {\em not popular but intriguing}. While in general mindsets are found to be helpful, the participants mentioned the lack of POI details as an inconvenience. One of our future directions is to enrich POI attributes by integrating exogenous data sources which describe different aspects of POIs. 

\vspace{-10pt}
\section{Related Work}
\label{sec:rel}
To the best of our knowledge, \sys is the first customized, contextualized and explainable POI recommendation system to tackle cold start to its full extent. 
However, our approach does relate to some other approaches in the literature, which we review in this section.

\separateshort
\noindent \textbf{Interactive recommendation.} Interactive recommendation is a crucial need for serving cold users and describing complex information needs~\cite{radlinski2017theoretical}. 
Each iteration of an interactive recommendation system constitutes a conversation, where the system provides some information to the user, and the user responds in the form of a feedback. Interactive recommendation is often modeled either as a two-phase approach (e.g., interview-based~\cite{DBLP:conf/iui/Amer-YahiaBEOV20}, onboarding~\cite{christakopoulou2018q}, active learning~\cite{DBLP:journals/tkde/DimitriadouPD16}) or as an exploratory data analysis (EDA) approach~\cite{omidvar2015interactive}. Most interactive models require users to either go through a tedious knowledge acquisition process, or deal with a complex feedback capturing mechanism. To the best of our knowledge, while \sys counts as an EDA approach, it is the first interactive POI recommendation method which provides a lightweight feedback capturing mechanism to guide users through the iterative conversations.

\separateshort
\noindent \textbf{Explainable recommendation.} The task of explainable recommendation is to inform users why they receive certain POIs as recommendation results.
Most explainability approaches focus on explaining latent representations in ML models in a non-intrusive way~\cite{zhang2014explicit,wang2018explainable}, where the internals of the systems remain unexplained. For an enhanced explainability, it is required to either feed some ontological a-priori knowledge to the process, or rebuild an explainable structure from the ground up. \sys is a transparency-by-design method where recommendations are naturally explainable using look-alike groups. While it is shown in~\cite{DBLP:conf/edbt/XanthopoulosTCS20} that feature explanation is among the most challenging aspects of explainable recommendation, look-alike groups easily manifest the set of features that made a POI set recommendable.

\separateshort
\noindent \textbf{Look-alike groups.} To tackle the challenge of cold start, many approaches are proposed to enrich a sparse check-in matrix with social aspects.
The assumption is that users may be more interested to visit places that their friends had visited in the past~\cite{10.1145/2525314.2525357}. In~\cite{ye2011exploiting}, a friendship-based CF is proposed. In~\cite{jamali2010matrix}, the Matrix Factorization model is extended with social regularization. Moreover, friendship links are exploited in~\cite{lalwani2015community} to build friend groups using Community Detection techniques. In \sys, we extend the domain of cold start to social aspects, i.e., no friendship link is available for users. \sys builds explainable look-alike groups without relying on any social aspects or check-ins of the user. 

\separateshort
\noindent \textbf{Mindsets.} Mindsets capture user needs. Category-based search interfaces~\cite{hearst2009search,10.1145/2666310.2666479} 
capture explicit needs of users in the form of categories (e.g., selecting POIs of the category ``historical landmarks''). Categories enable more personalization for users, resulting in less anxiety and more trust over the system. However, realistic scenarios often contain ambiguous needs and intents, where users seek to disambiguate in an iterative process. The only possible iteration in traditional search paradigms is to restart a search with another category. Exploratory travel interfaces~\cite{Viswanathan:2019:ACN:3301019.3323885}
have been found to enhance user experience in POI exploration with serendipity measures.
In \sys, we employ ``mindsets'', which is an intuitive way of capturing user's implicit intents. Mindsets reflect interests of users, and align to users' preferences during iterations.

\vspace{-5pt}
\section{Conclusion}
\label{sec:conc}
We present \sys, an interactive and explainable POI recommendation system based on look-alike groups, which tackles the common challenges of cold start, customizability, contextuality, and explainability. We introduce the notion of ``mindsets'' which extends user context, and captures actual situation and intents of the user. In an extensive set of experiments, we show that \sys achieves a Hit Ratio higher than $50\%$ only after $10$ iterations. 
In an extensive user study, we showed the effectiveness of look-alike groups and mindsets for POI recommendation.

We have several directions to pursue as future work. We plan to integrate planning and mobility aspects in \sys. While our focus in the current paper was on the actual context of the user, the functionality of the system can be enriched with some user-defined constraints on POI preferences and mobility patterns. 

\section*{Acknowledgement}
We thank Antonietta Grasso, Adrien Bruyat, C\'ecile Boulard, and Denys Proux, for their valuable comments and feedback.

\bibliography{references}

\section*{Appendix}
\appendix

\section{NP-Completeness Proof}
\label{apx:np}

\begin{theorem}
The POI recommendation problem defined in Section~\ref{sec:problemdef} is NP-Complete.
\end{theorem}

\begin{proof}
(sketch) It is shown in~\cite{DBLP:journals/vldb/Omidvar-Tehrani19} that the problem of picking $k$ groups by maximizing an optimization objective is NP-Complete by a reduction from Maximum Edge Subgraph (MES) problem. In the problem of POI recommendation using look-alike groups, we have two additional elements: $(i)$ relevance and distance constraints, and $(ii)$ maximization of more than one objective (as mindset functions combine several utility functions as objectives). This means that the problem in~\cite{DBLP:journals/vldb/Omidvar-Tehrani19} is a special case of ours. Hence our problem is obviously harder.
\end{proof}

\section{Transaction embeddings}
\label{apx:embeddings}

Transaction embeddings are encoded as integers. Hence all POIs and demographic values should be encoded to single integers. While POIs are often already associated to a unique identifier which can be directly used in the embeddings, demographic attributes may have different values in discrete or continuous domains. For instance, the demographic attribute ``number of check-ins'' may contain different values in a wide range. We employ an equal-frequency discretization approach to obtain $4$ categories for each demographic attribute.
Table~\ref{tbl:discretization} illustrates discretization values for demographic attributes in \textsc{Gowalla}. The discretized attributes result in $28$ items ($7$ attributes and $4$ categories for each) to be concatenated to the embeddings. For instance, a visitor embedding may contain items ``many places'' and ``few photos'', which means that the visitor went to many places, but did not take many pictures.

\begin{table}[h]
\begin{tabular}{|l|l|l|l|l|}
\hline
         & \textbf{very few}          & \textbf{few}       & \textbf{some}       & \textbf{many}       \\ \hline \hline
\textbf{items}    & $\leq 2$    & $(2,3]$   & $(3,5]$    & $\geq 5$ \\ \hline
\textbf{photos}   & $\leq 1$     & $(1,2]$   & $(2,5]$    & $\geq 5$ \\ \hline
\textbf{friends}  & $\leq 1$     & $(1,3]$   & $(3,5]$    & $\geq 5$ \\ \hline
\textbf{check-ins} & $\leq 3$ & $(3, 12]$ & $(12, 34]$ & $\geq 34$ \\ \hline
\textbf{places}   & $\leq 3$ & $(3,9]$   & $(9,23]$   & $\geq 23$ \\ \hline
\end{tabular}
\caption{Equal-frequency discretization of demographic attributes in \textsc{Gowalla} dataset.}
\label{tbl:discretization}
\vspace{-20pt}
\end{table}

Beyond demographics and visited POIs, we enrich the embeddings by POI categories and check-in time. 
For a POI $p$ in an embedding of a visitor $u$, we concatenate $p.\mathit{att}$ to $\mathit{trans}(u)$. For instance, if \textit{Louvre} $\in \mathit{trans}(u)$, we also consider $\langle \mathit{cat}, \mathit{museum}\rangle$ as an additional item in $\mathit{trans}(u)$. This enables the system to generate a group of visitors who check in museums in general, but not necessarily Louvre. For instance in Figure~\ref{fig:example}, the visitors of the green group have all checked in a ``historical landmark'', but it could be different landmarks that they had checked in. Also for a pair of POI and visit time $\langle p,t \rangle \in u.\mathit{checkins}$, we discretize $t$ to {\em hourly} 
categories (``morning'', ``afternoon'', ``evening'', and ``night'') 
and {\em weekly} categories 
(``weekday'' and ``weekend'') 
and concatenate $\mathit{trans}(u)$ with those hourly and weekly categories alongside $p$'s category. For instance in Figure~\ref{fig:example2}, members of the green group have all checked in a restaurant (any POI with the category ``restaurant'') in the evening.

\section{Performance Study}
\label{apx:perf}
In addition to the simulation study (Section~\ref{sec:simstudy}) and the user study (Section~\ref{sec:userstudy}), we evaluate the efficiency of \sys by measuring the average execution time at each iteration. We conjecture that the radius $r$ and the number of groups $k$ are the two most influencing input parameters on the performance of \sys. Hence we report the execution time by varying $r$ between $50m$ and $1km$, and $k$ between $5$ and $70$. Note that we don't analyze the $k'$ parameter, as it is a dependent variable to $k$. All the performance experiments are conducted on an 2.2 GHz Intel Core i7 with 32GB of DDR4 memory on OS X 10.14.6 operating system.

\begin{figure}[h]
    \centering
    \includegraphics[width=\columnwidth]{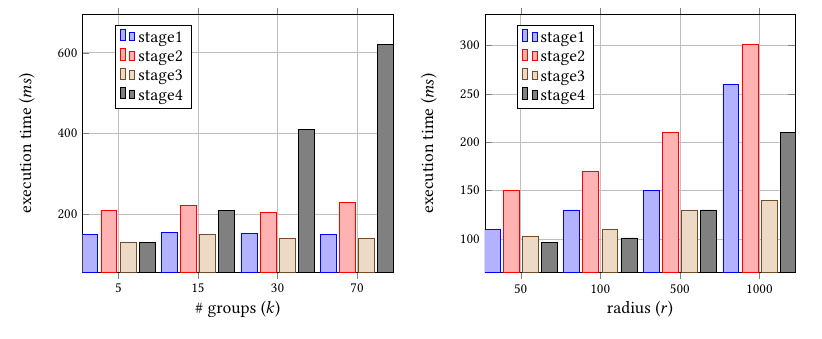}
    \caption{Execution time of \sys.}
    \label{fig:performance}
\end{figure}

\begin{figure*}[t]
    \centering
    \includegraphics[width=\textwidth]{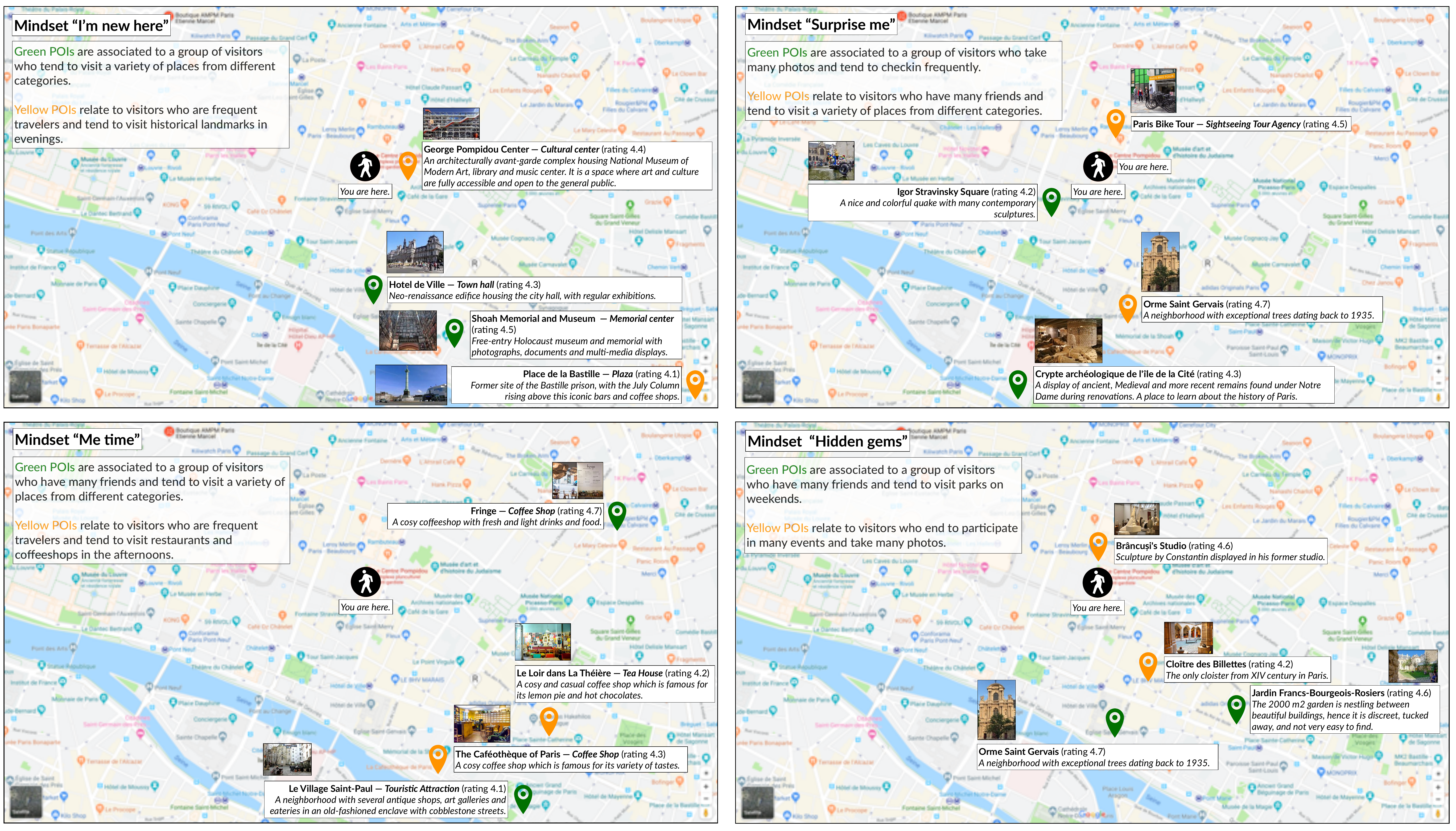}
    \caption{An example of POIs provided in the user study.}
    \label{fig:amtexample}
\end{figure*}

To dissect \sys's performance, we report the execution time for each of the following stages of Algorithm~~\ref{algo:iprelo}: {\em stage~1:} find nearby POIs ($\mathit{nearby\_POIs}()$ function in line~\ref{ln:getpoi} of the algorithm), {\em stage 2:} find check-ins of the nearby POIs ($\mathit{checkins\_of}()$ function in line~\ref{ln:getcheckins}), {\em stage 3:} build the transaction matrix and generate look-alike groups ($\mathit{mine\_groups}()$ function in line~\ref{ln:minegroups}), and {\em stage 4:} find $k$ optimal groups and their $k'$ associated POIs regarding a given mindset ($\mathit{maximize}()$ function in line~\ref{ln:max}). We report the average execution time over $100$ iterations. Each iteration is simulated as described in Section~\ref{sec:simstudy}. However note that we don't analyze sessions in the performance study, and we consider that each iteration is independent from others. By default, we consider $k=k'=5$, and $r=0.5km$.

Figure~\ref{fig:performance} depicts the results of our performance study. The left and right charts report execution time by varying $k$ and~$r$, respectively. We observe that increasing $k$ only influences stage 4 (i.e., mindset maximization), and the execution time of all other stages stays constant. The average execution time for the three first stages are $152ms$, $216.25ms$, and $140ms$, respectively. The time for executing the fourth stage grows from $130ms$ at $k=5$ to $621.75ms$ at $k=70$. Note that while the time limit $tl$ is always fixed to $100ms$ (see Section~\ref{sec:algo}), stage 4 requires extra time to obtain $k'$ POIs for each group, generate output, and save it on disk. While~$k$ is oblivious to the three first stages of \sys, increasing $r$ influences the execution time of all the stages. The reason is that a larger radius implies that more POIs and consequently more check-ins should be retrieved from the database. Abundance of check-ins makes the group candidate space larger, hence the group generation stage requires more time. The overall execution time of \sys grows from $459ms$ at $r=50m$ to $911ms$ at $r=1km$. We consider $r=500m$ as the default value whose collective execution time is on average $620ms$.


\section{User study}
\label{apx:userstudy}
We recruited $753$ participants in AMT 
and forwarded them to a Survey Monkey questionnaire\footnote{\it https://www.surveymonkey.com} to answer different questions about the functionality of our proposed system. $51\%$ of the participants were female. Also the majority of them was in the age range of 25-34 ($44.09\%$) followed by the age ranges 34-44 and 18-24 by $20.22\%$ and $18.04\%$, respectively. Moreover, $54\%$ of them were from the US, $22\%$ from India, and $24\%$ from the rest of the world. Each participant received $\$0.07$ as the incentive to complete the study. To make the results of our study easily interpretable and comparable, we consider a controlled environment where the participants always select one of the mindsets in Table~\ref{tbl:mindsets} and won't propose a new mindset (mindset creation is discussed in Section~\ref{sec:beyond}). The user study for mindset creation is a part of our future work.

\separateshort
\noindent \textbf{Before the user study.} Before we show \sys results to the participants, we perform a pre-test to identify motivations behind using a POI recommender. The participants should select at least one among $5$ pre-defined options depicted in Table~\ref{tbl:motivation}. We observe that the dominant motivations are ``knowing more about places'' and ``spending time with others''. The exploratory nature of \sys helps users to know better their POI options in several iterations. Also look-alike groups help them to receive an explanation for the recommendations, and enable them to decide better about where they can go with their company. The third dominant option is ``well-being'', which stresses on the specificity of users' intents. Mindsets in \sys are designed to capture such intents. For instance, by selecting the mindsets $m_3$ (let's workout) and $m_4$ (me-time), the user will receive recommendations which optimize his/her physical and mental well-being, respectively.

\begin{table}[h]
\begin{tabular}{|l|l|}
\hline
\textbf{Option}                      & \textbf{Vote (\%)} \\ \hline \hline
know more about places               & 58.38             \\ \hline
spend time with friends and family   & 56.48             \\ \hline
well-being (mental and physical)      & 37.92             \\ \hline
post stories and photos on Instagram & 20.71             \\ \hline
boredom                              & 16.96             \\ \hline
\end{tabular}
\caption{Motivations for using a POI recommender}
\label{tbl:motivation}
\vspace{-20pt}
\end{table}

\separateshort
\noindent \textbf{During the user study.} The pre-test helped us design AMT tasks for our user study based on users' needs. 
Given a mindset $m$, the participants receive $k$ groups and $k'$ POIs optimized for $m$. In comparative studies, $k \times k'$ POIs of a baseline will be also illustrated side-by-side. Figure~\ref{fig:amtexample} shows examples of POIs that the participants received during our user study, for the mindsets ``I'm new here'', ``surprise me'', ``me time'', and ``hidden gems'' ($k=k'=2$). All POIs are at most one kilometer far (i.e., $r=1km$) from George Pompidou Center in Paris (i.e., the user's location). As AMT participants may not know Paris, each POI is annotated with a photo, rating, and a concise description, using Google Places API\footnote{\it cloud.google.com/maps-platform/places}.

\end{document}